\newif\ifniceformat
\newif\iflcsscolor
\niceformattrue     
\lcsscolorfalse

\ifniceformat
\documentclass[11pt]{article}
\usepackage{textcomp}
\usepackage{amsmath}
\usepackage{amssymb}
\usepackage{amsthm}
\usepackage{todonotes}
\usepackage{tikz}
\usepackage{pgfplots}
\usepackage{graphicx}
\usepackage{epstopdf}
\usepackage{wrapfig}
\usepackage[labelfont=bf,font=footnotesize]{caption}
\usepackage[labelfont=bf,font=footnotesize]{subcaption}
\usepackage{xcolor}
\usepackage{bbm}
\usepackage{makecell}
\usepackage{authblk}
\usepackage{fullpage,etoolbox}
\usepackage[shortlabels]{enumitem}
\usepackage[round,sort,compress]{natbib}
\usepackage{listings}
\usepackage{auxFiles/commands}
\usepackage[ruled,vlined,linesnumbered]{algorithm2e}
\usepackage{todonotes}
\usepackage{float}
\usepackage{cancel}
\usepackage{makecell}
\usepackage{balance}  
\usepackage[colorlinks=true,citecolor=blue,linkcolor=blue,urlcolor=blue]{hyperref}

\usepackage{auxFiles/symbolDef}

\makeatletter
    \@ifpackageloaded{xcolor}{}{\usepackage{xcolor}}
\makeatother
\makeatletter
    \@ifpackageloaded{needspace}{}{\usepackage{needspace}}
\makeatother

\newcommand{\uppercaseGreek}[1]{%
    \begingroup
    \ucmathlist\MakeUppercase{#1}%
    \endgroup
}

\newcommand{\ucmathlist}{%
    \def\alpha{\mathrm{A}}%
    \def\beta{\mathrm{B}}%
    \let\gamma=\Gamma
    \let\delta=\Delta
    \def\epsilon{\mathrm{E}}%
    \def\varepsilon{\mathrm{E}}%
    \def\zeta{\mathrm{Z}}%
    \def\eta{\mathrm{H}}%
    \let\theta=\Theta
    \let\vartheta=\Theta
    \def\iota{\mathrm{I}}%
    \def\kappa{\mathrm{K}}%
    \let\lambda=\Lambda
    \def\mu{\mathrm{M}}%
    \def\nu{\mathrm{N}}%
    \let\xi=\Xi
    \let\pi=\Pi
    \let\varpi=\Pi
    \def\rho{\mathrm{P}}%
    \def\varrho{\mathrm{P}}%
    \let\sigma=\Sigma
    \def\tau{\mathrm{T}}%
    \let\upsilon=\Upsilon
    \let\phi=\Phi
    \let\varphi=\Phi
    \def\chi{\mathrm{X}}%
    \let\psi=\Psi
    \let\omega=\Omega
}

\makeatletter
    \@ifpackageloaded{amsthm}{}{

        \usepackage{amsthm}
    }
\makeatother
\theoremstyle{plain}
    \newtheorem{theorem}{Theorem}

    \newtheorem{corollary}[theorem]{Corollary}
\theoremstyle{definition}
    \newtheorem{definition}{Definition}
    \newtheorem{remark}{Remark}
    \newtheorem{assumption}{Assumption}

\makeatletter
\def\renewtheorem#1{%
    \expandafter\let\csname#1\endcsname\relax
    \expandafter\let\csname c@#1\endcsname\relax
    \gdef\renewtheorem@envname{#1}
    \renewtheorem@secpar
}
\def\renewtheorem@secpar{\@ifnextchar[{\renewtheorem@numberedlike}{\renewtheorem@nonumberedlike}}
\def\renewtheorem@numberedlike[#1]#2{\newtheorem{\renewtheorem@envname}[#1]{#2}}
\def\renewtheorem@nonumberedlike#1{  
    \def\renewtheorem@caption{#1}
    \edef\renewtheorem@nowithin{\noexpand\newtheorem{\renewtheorem@envname}{\renewtheorem@caption}}
    \renewtheorem@thirdpar
}
\def\renewtheorem@thirdpar{\@ifnextchar[{\renewtheorem@within}{\renewtheorem@nowithin}}
\def\renewtheorem@within[#1]{\renewtheorem@nowithin[#1]}
\makeatother

\else

\iflcsscolor
\documentclass[journal,twoside,web]{ieeecolor}
\usepackage{xcolor}
\usepackage{auxFiles/lcsys}
\else
\documentclass[letterpaper, 10 pt, conference]{ieeeconf}
\IEEEoverridecommandlockouts 
\usepackage{xcolor}
\fi
\usepackage[sort,compress]{cite}
\usepackage{graphicx}
\usepackage{amsmath, amssymb}
\usepackage{auxFiles/commands}
\usepackage{csquotes}
\usepackage[ruled,vlined,linesnumbered]{algorithm2e}
\usepackage{float}
\usepackage{todonotes}
\usepackage{subcaption}
\usepackage{makecell}
\usepackage{cancel}
\usepackage[normalem]{ulem}
\newcommand{\citep}[1]{\cite{#1}}
\newcommand{\citet}[1]{\cite{#1}}

\usepackage{auxFiles/symbolDef}

\makeatletter
\let\NAT@parse\undefined
\makeatother
\usepackage[colorlinks=true, allcolors=blue]{hyperref}
\fi
\ifniceformat
    \title{Scalable Distributed Nonlinear Control \\ Under Flatness-Preserving Coupling}
\else
    \iflcsscolor
        \title{Scalable Distributed Nonlinear Control \\ Under Flatness-Preserving Coupling}
    \else
        \title{\LARGE \bf Scalable Distributed Nonlinear Control \\ Under Flatness-Preserving Coupling}
    \fi
\fi

\iflcsscolor
    \author{Fengjun Yang, \IEEEmembership{Student Member, IEEE}, Jake Welde, and Nikolai Matni, \IEEEmembership{Member, IEEE}
    \thanks{Manuscript submitted for review September 12, 2025. This work was supported in part by NSF Awards SLES-2331880, ECCS-2045834, and ECCS-2231349 and AFOSR Award FA9550-24-1-0102.}
    \thanks{F. Yang and N. Matni are with the GRASP Lab at the University of Pennsylvania, Philadelphia, PA 19104. (e-mails: \{fengjun, nmatni\}@seas.upenn.edu).}
    \thanks{J. Welde is with the Sibley School of Mechanical and Aerospace Engineering, Cornell University, Ithaca, NY 14850.
    (e-mail: jakewelde@cornell.edu).}}
\else
    \author{Fengjun Yang, Jake Welde, and Nikolai Matni
    \thanks{F. Yang and N. Matni are with the GRASP Laboratory, University of Pennsylvania, PA, USA. 
    J. Welde is with the Sibley School of Mechanical and Aerospace Engineering, Cornell University, Ithaca, NY, USA.
    This work was supported in part by NSF Awards SLES-2331880, ECCS-2045834, ECCS-2231349, AFOSR Award FA9550-24-1-0102, and DARPA TRS program under contract HR00112590145.}%
    }
\fi

\begin{document}
\maketitle

\begin{abstract}%
We study distributed control for a network of nonlinear, differentially flat subsystems subject to dynamic coupling. Although differential flatness simplifies planning and control for isolated subsystems, the presence of coupling can destroy this property for the overall joint system. Focusing on subsystems in pure-feedback form, we identify a class of compatible lower-triangular dynamic couplings that preserve flatness and guarantee that the flat outputs of the subsystems remain the flat outputs of the coupled system. Further, we show that the joint flatness diffeomorphism can be constructed from those of the individual subsystems and, crucially, its sparsity structure reflects that of the coupling. Exploiting this structure, we synthesize a distributed tracking controller that computes control actions from local information only, thereby ensuring scalability. We validate our proposed framework on a simulated example of planar quadrotors dynamically coupled via aerodynamic downwash, and show that the distributed controller achieves accurate trajectory tracking.

\end{abstract}

\section{Introduction}
The distributed control of multi-agent systems---a central problem in the deployment of large teams of robots---is challenging due to dynamic coupling between subsystems and limited communication between controllers. These challenges are further exacerbated when the system dynamics are nonlinear. To deal with this challenge, existing approaches often require restrictive structural assumptions or suffer from high computation costs \citep{hill1980dissipative, scattolini2009architectures}. For many nonlinear systems, differential flatness offers a powerful framework for planning and control \citep{fliess1995flatness} by establishing \textit{flat outputs} and a \textit{flatness diffeomorphism} that maps between the system and an equivalent linear one, thereby simplifying planning and control. However, the application of flatness-based approaches to the control of dynamically coupled multi-agent systems remains limited, as the presence of such coupling between subsystems may destroy flatness, even when the subsystems (in isolation) are themselves differentially flat.

To overcome such limitations, we build upon our prior work \citep{yang2025learning} to show that for flat subsystems in pure feedback form, a class of lower-triangular couplings between subsystems preserves differential flatness of the joint system. Moreover, the joint flat output is simply the concatenation of those of the individual subsystems, giving rise to a flat space where the subsystems evolve under decoupled, linear dynamics. We further identify mild assumptions under which the sparsity structure of the joint flatness diffeomorphism echoes that of the dynamic coupling. Such sparsity enables the design of a distributed flatness-based tracking controller for the joint system, since only local information is needed to evaluate each subsystem's component of the joint flatness diffeomorphism. We validate our approach in a simulation, where a team of planar quadrotors operates in proximity (and thus subject to downwash coupling). Our results show that, despite the coupled dynamics, the proposed controller achieves accurate tracking under limited communication.

\subsection{Related Work}

\paragraph{Nonlinear Distributed Control}
Existing literature on nonlinear distributed control spans several paradigms. Distributed nonlinear model predictive control (DNMPC) \citep{scattolini2009architectures, muller2017economic, van2017distributed} solves constrained optimal control problems via distributed optimization and can accommodate arbitrary subsystem dynamics. However, such methods suffer from high computation costs, limiting their application in real-time to systems with fast dynamics. Small-gain methods \citep{jiang1994small,dashkovskiy2010small} provide robust stability guarantees by requiring the subsystems to be input-to-state stable with respect to their interconnections, which can lead to conservatism in practice. Alternatively, passivity-based methods \citep{hill1980dissipative} also guarantee stability but are only applicable to a restrictive class of energy-dissipative couplings. In contrast, our method leverages differential flatness to reduce a nonlinear distributed control problem to a linear one. For systems conforming to our assumptions, we avoid the computational burden of DNMPC and the conservatism of small-gain methods, while tolerating couplings that might not be strictly dissipative.

\paragraph{Distributed Control of Flat Systems}
Differential flatness has long been exploited in the control of single-agent systems, as well as in consensus and formation control for multi-agent systems whose subsystems are completely decoupled \citep{mondal2022consensus, mondal2025distributed, van2017distributed}. However, for dynamically coupled flat subsystems, flatness of the joint system is not guaranteed, nor is it clear how to construct the joint diffeomorphism in general \citep{nicolau2022flatness}.
\ifniceformat
\citet{menara2020conditions}
\fi
provide a sufficient condition for the flatness of complex networks, but require solving a partial differential equation for the flat outputs.
\ifniceformat
\citet{bidram2013distributed, bidram2014synchronization}
\fi
apply the closely-related notion of feedback linearization to a network of dynamically coupled power generators, but assume that the coupling is quasi-static and thus can be approximated as a constant in the diffeomorphism. In contrast, focusing on pure-feedback systems, we identify a class of dynamic couplings where the joint system is provably flat and explicitly construct its flat outputs and diffeomorphism, without assuming quasi-static coupling.

\subsection{Contributions}
Our contributions are two-fold.
\begin{itemize}
    \item \textbf{Flatness-Preserving Coupling}: We identify a class of lower-triangular dynamic coupling that preserves flatness for subsystems in pure-feedback form. Further, we show that the flatness diffeomorphism of the joint system can be constructed explicitly, with a sparsity structure echoing that of the dynamic coupling.
    \item \textbf{Distributed Flatness-based Controller}: Exploiting the strong structural properties of the flatness diffeomorphism, we design a distributed tracking controller, wherein each subsystem only requires local information to compute its control input. We validate the proposed framework in simulation with planar quadrotors coupled via aerodynamic downwash, demonstrating accurate tracking under limited communication.
\end{itemize}

\subsection{Notation}
For a signal $\vcx(t)$, we denote the $r$-th order time derivative of $\vcx$ as $\vcx^{(r)}$. For the first and second-order time derivatives of $\vcx$, we also use the shorthand $\dot\vcx$ and $\ddot\vcx$, respectively. For a positive integer $N$, we define $[N]:=\{1, \ldots, N\}$ to be the set of positive integers smaller than or equal to $N$. We use $D_\vcz h(\cdot)$ to denote the partial Jacobian matrix of a function $h$ with respect to $\vcz$ and $Dh(\cdot)$ to denote the full Jacobian of $h$ with respect to all of its arguments. For derivatives of higher orders, we use $D_\vcz^k h(\cdot)$ to denote the $k$-th order Fréchet derivative of $h$ with respect to $\vcz$.
\section{Problem Formulation}
Consider a network of $N$ nonlinear systems, where each subsystem $i$ has state $\vcx^i \in \R^{n_i}$ and control input $\vcu^i \in \R^{m_i}$. Denote the joint state and control input as
$$\vcx := [\vcx^1, \ldots, \vcx^N], \quad \vcu := [\vcu^1, \ldots, \vcu^N].$$
The joint system evolves under the smooth dynamics
\begin{equation}\label{eq: joint-dynamics}
    \dot\vcx = f(\vcx,\vcu).
\end{equation}
We assume that the joint dynamics can be decomposed as
\begin{equation}\label{eq: coupled-joint-dynamics}
    f(\vcx, \vcu) = \bar{f}(\vcx, \vcu) + \Delta(\vcx),
\end{equation}
where $\bar{f}(\vcx, \vcu)$ captures the uncoupled subsystem dynamics,
\begin{equation}\label{eq: uncoupled-dynamics}
    \bar{f}(\vcx, \vcu) = \begin{bmatrix}
        \bar{f}^1(\vcx^1, \vcu^1) \\
        \vdots \\
        \bar{f}^N(\vcx^N, \vcu^N)
    \end{bmatrix},
\end{equation}
and $\Delta(\vcx)$ represents dynamic coupling between the systems (e.g., aerodynamic interference or physical interconnections).

We consider uncoupled subsystem dynamics $\bar{f}^i$ that are \textit{differentially flat}\footnote{For a detailed treatment of differential flatness, see e.g., \citet{fliess1995flatness}.}, that is, there exists a \textit{flat output}
\begin{equation*}
    \vcy^i = \Lambda^i(\vcx^i, \vcu^i, \dot{\vcu^i}, \ddot{\vcu^i}, \cdots, \vcu^{i^{(\alpha)}})
\end{equation*}
such that a finite number of its time derivatives can be used to reconstruct the state $\vcx^i$ and control input $\vcu^i$:
\begin{align*}
    \vcx^i &= \Phi^i(\vcy^i, \dot{\vcy^i}, \ddot{\vcy^i}, \dots, \vcy^{i^{(\beta)}}),\\
    \vcu^i &= \Psi^i(\vcy^i, \dot{\vcy^i}, \ddot{\vcy^i}, \dots, \vcy^{i^{(\beta)}}).
\end{align*}
Here, $\Lambda^i$, $\Phi^i$, and $\Psi^i$ are smooth functions, with $(\Phi^i,\Psi^i)$ referred to as the \textit{flatness diffeomorphism}. The quantities $\alpha$ and $\beta$ are finite positive integers.

We seek to establish and leverage the differential flatness of the joint system to simplify controller design. However, even if the uncoupled subsystems are differentially flat, it is unclear what structural assumptions need to be imposed on $\Delta$ to render the joint system flat. Moreover, flatness alone does not provide a straightforward path towards a control strategy amenable to distributed implementation, since the maps to and from the flat space may in general depend on arbitrary subsystem states. We therefore seek to understand when the joint flatness diffeomorphism exhibits beneficial sparsity.

To address these challenges, we focus on subsystems with uncoupled dynamics that can be placed in pure-feedback form. Such systems are differentially flat under mild regularity conditions\footnote{We defer the details on the regularity condition to Assumption~\ref{assm: uncoupled-dynamics-regular}.}, and the dynamics of many mechanical systems can be expressed in this form ({e.g.}, unicycles, planar quadrotors, and fully-actuated manipulators) \citep{rathinam1998configuration}. Further, their rich structures allow us to characterize a class of compatible dynamic couplings with desired properties. Formally, we assume that the subsystem dynamics are pure-feedback with relative degree $r$, i.e., that there exists a decomposition $\vcx^i = [\vcx^i_1, \ldots, \vcx^i_r]$ such that the (decoupled) dynamics of subsystem $i$ take the form\footnote{For simplicity, we assume that each subsystem has relative degree $r$ and input dimension $m_i$, so that $\vcx^i \in \R^{rm_i}$. The analysis can be extended to systems with heterogeneous relative degrees via dynamic extension \citep{yang2025learning}.}
\begin{equation}\label{eq: flat-PFF}
    \dot\vcx^i = \bar f^i(\vcx^i, \vcu^i)
    =
    \begin{bmatrix}
        \dot \vcx^i_1 \\
        \dot \vcx^i_2 \\
        \vdots \\
        \dot \vcx^i_r
    \end{bmatrix}
    =
    \begin{bmatrix}
        \bar{f}^i_1(\vcx^i_1, \vcx^i_2) \\
        \bar{f}^i_2(\vcx^i_1, \vcx^i_2, \vcx^i_3) \\
        \vdots \\
        \bar{f}^i_r(\vcx^i_1, \ldots, \vcx^i_r, \vcu^i)
    \end{bmatrix}.
\end{equation}

\begin{remark}
    We use superscripts to index subsystems and subscripts to index the pure-feedback order. For instance, $\vcx^i_j$ denotes the $j^\textrm{th}$-order substate of the $i^\textrm{th}$ subsystem. When convenient, we also write $\vcx_j = [\vcx^1_j, \ldots, \vcx^N_j]$ to denote the collection of $j^\textrm{th}$-order substates across all subsystems.
\end{remark}

Our goal is to identify a constraint on the dynamic coupling term $\Delta$ under which the joint system \eqref{eq: joint-dynamics} is differentially flat, with flat outputs and a flatness diffeomorphism that are structurally compatible with distributed controller implementation. To this end, we first identify a class of \textit{lower-triangular} couplings that ensures the joint dynamics are differentially flat (Theorem~\ref{thm: flatness-of-canonical-coupling}), with flat outputs identical to those of the uncoupled subsystems and therefore locally computable. We then show that the joint diffeomorphism $(\Phi, \Psi)$ can be constructed from those of the subsystems (Theorem~\ref{thm: form-of-pert-flat-map}). Finally, we analyze the structural properties of $(\Phi, \Psi)$ (Corollary~\ref{cor: coupling-locality}) and show that they can be exploited to design distributed tracking controllers in Section~\ref{sec: tracking-control}.



\section{Flatness-Preserving Coupling}\label{sec: fpp}
We consider a class of \textit{lower-triangular} couplings that preserve the pure-feedback structure of the uncoupled dynamics. We show that under suitable regularity conditions, the joint system is differentially flat under this class of dynamic coupling, and the joint flat output is the collection of those of the uncoupled systems. Further, the flatness diffeomorphism of the joint system under this class of couplings can be constructed from those of the uncoupled systems. We first recall a standard regularity assumption on the uncoupled dynamics \eqref{eq: uncoupled-dynamics}.
\begin{assumption}[Regularity of Uncoupled Dynamics]
\label{assm: uncoupled-dynamics-regular}
    For each subsystem $i=1,\dots,N$, there exists a neighborhood of the operating point $(\vcx^*, \vcu^*)$ on which the uncoupled dynamics \eqref{eq: flat-PFF} satisfy
    \begin{itemize}
        \item the $\bar{f}^i_j$ are smooth, and
        \item the partial Jacobians 
            \begin{align*}
            \left\lvert D_{\vcx_{j+1}}\bar{f}^i_j(\vcx^*_1, \ldots,\vcx^*_{j+1}) \right\rvert &\neq 0, \quad j=1, \ldots, r-1,\\ \left\lvert D_{\vcu}\bar{f}^i_r(\vcx^*_1, \ldots,\vcx^*_{r}, \vcu^*) \right\rvert &\neq 0,
            \end{align*}
            are non-singular.
        \end{itemize}
\end{assumption}
It is well-known \citep{murray1995differential} that for systems satisfying Assumption~\ref{assm: uncoupled-dynamics-regular}, there exists an open neighborhood $\stX \ni (\vcx^*,\vcu^*)$ on which all subsystems are differentially flat under the flat outputs $\vcy^i=\vcx^i_1$, respectively.\footnote{See our prior work \citet{yang2025learning} for a proof of this fact in this notation.}

We denote the coupling term affecting subsystem $i$ as $\Delta^i$, and further, its effect on the $j^\textrm{th}$ substate as $\Delta^i_j$, such that
\begin{equation}
    \dot{\mathbf{x}}^i_j = \bar{f}^i_j(\mathbf{x}^i_1, \mathbf{x}^i_2, \ldots, \mathbf{x}^i_{i+1}) + \Delta^i_j(\mathbf{x}).
\end{equation}
Building on our prior work \citet{yang2025learning}, we consider a class of \textit{lower-triangular} couplings that preserves the pure-feedback structure of the joint dynamics and allows for the flatness diffeomorphism of the joint system to be found via a forward-substitution-like recursion.

\begin{definition}[Lower-Triangular Coupling]\label{defn: lower-triangular-coupling}
    A coupling term ${\Delta}(\vcx)$ is \textit{lower-triangular} on $\stX$ if its effect on the $j^\textrm{th}$ substate of the subsystem $i$, $\Delta^i_j$, is smooth and can be expressed as
    \begin{equation}\label{eq: defn-canonical-lt-coupling}
        \Delta^i_j(\vcx) = \bar{\Delta}^i_j(
        \underbrace{\vcx_1, \ldots, \vcx_{j}}_{=:\vcx_{\leq j}},
        \underbrace{\vcx^1_{j+1}, \ldots, \vcx^{i-1}_{j+1}}_{=:\vcx^{<i}_{j+1}}),
    \end{equation}
    for all $(\vcx, \vcu) \in \stX$. In other words, $\Delta^i_j$ depends only on terms of a lower feedback order $\vcx_{\leq j}$ and $(j+1)^\textrm{th}$ order terms from subsystems with lower indices $\vcx^{<i}_{j+1}$.
\end{definition}

\begin{figure}
    \centering
    \includegraphics[width=0.7\linewidth]{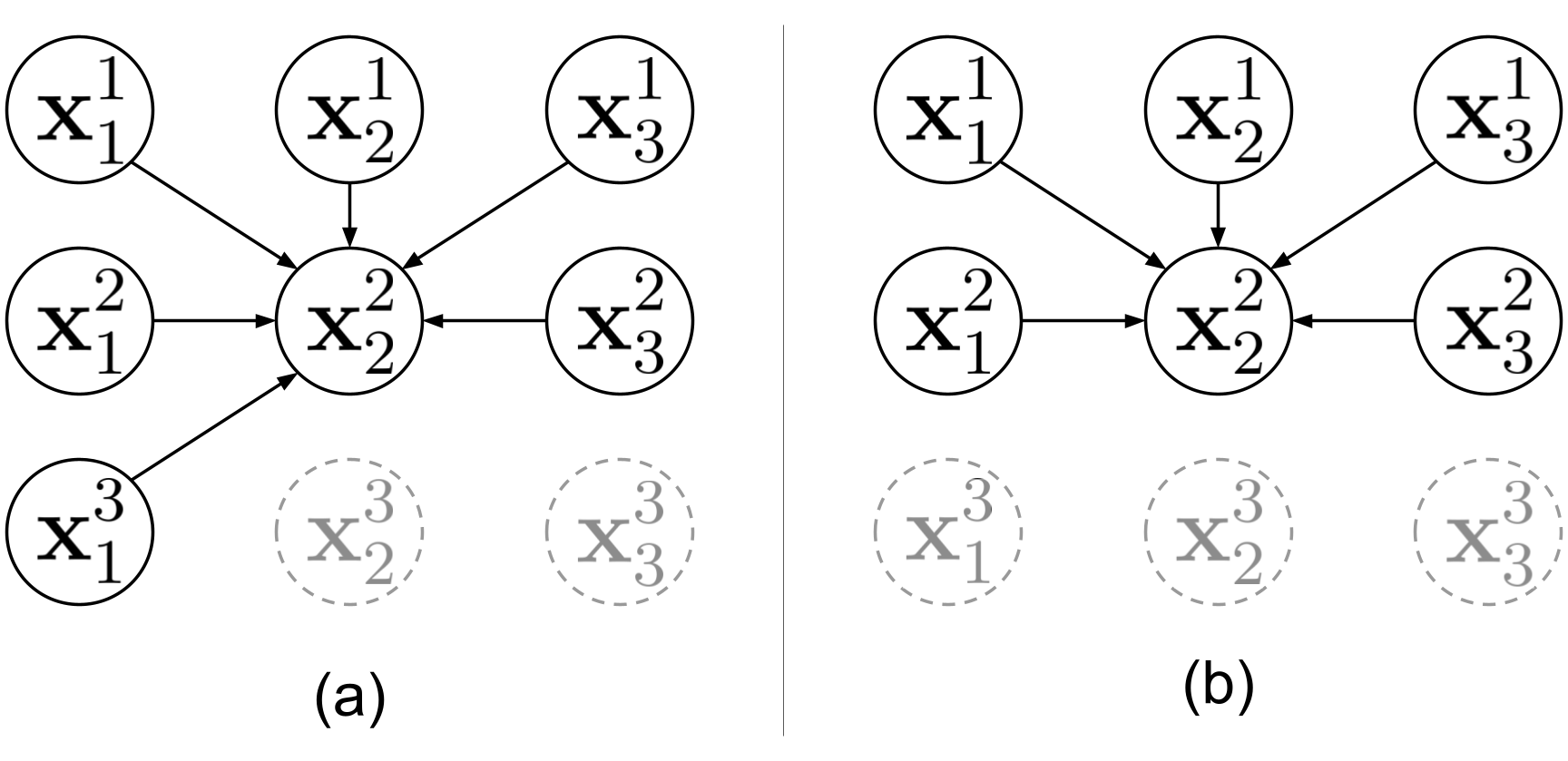}
    \caption{Illustration of the allowed dependencies for $\Delta^2_2$ on $N=3$ subsystems with relative degree $r=3$. (a) Lower-triangular coupling (Def.~\ref{defn: lower-triangular-coupling}) (b) Strongly lower-triangular coupling (Def.~\ref{defn: block-lower-triangular})}
    \label{fig: lt-illustration}
\end{figure}

While lower-order states from any subsystem may affect the $i$-th subsystem's states of the same or higher order, the dependence on $\vcx^{<i}_{j+1}$ in \eqref{eq: defn-canonical-lt-coupling} implicitly imposes an acyclic structure: higher-indexed subsystems may be affected by $(j+1)^\textrm{th}$ order states from lower-indexed ones, but not vice versa. \rewrite{}{The lower-triangular structure allows for a cascade decomposition of the joint system, which is closely related to the existence of invariant foliations on the state manifold \citep{fliess1985decompositions}}. For simplicity, we present Definition~\ref{defn: lower-triangular-coupling} and the subsequent results for the canonical case where this acyclic ordering is fixed across $\stX$ and coincides with the subsystem indices. \rewrite{}{We leave extending our results to a case where the ordering is state-varying to future work.}

We emphasize that the acyclicity restriction only applies to $(j+1)^\textrm{th}$ order dependencies. In Section~\ref{sec: locality-of-diffeo}, we analyze a special case\rewrite{}{, named \textit{strongly lower-triangular couplings} (Def.~\ref{defn: block-lower-triangular}),} where such dependencies are absent \rewrite{and thus no acyclicity is required}{. In this case, no acyclic ordering is required, allowing the framework to accommodate bidirectional and cyclic couplings}. \rewrite{}{While the lower-triangular assumption does not cover all arbitrary coupling structures, it captures highly relevant physical interactions, e.g., networks of quadrotors dynamically coupled via downwash, where empirically, the dominant aerodynamic effects depend only on the relative position of the vehicles \citep{shi2021neural}.} In Figure~\ref{fig: lt-illustration}, we show an illustration of these structural assumptions for a specific choice of $i$ and $j$.

\begin{theorem}\label{thm: flatness-of-canonical-coupling}
Let Assumption~\ref{assm: uncoupled-dynamics-regular} hold for the uncoupled dynamics and suppose that the coupling term $\Delta(\vcx)$ is lower-triangular. Then, the joint dynamics \eqref{eq: coupled-joint-dynamics} are differentially flat on $\stX$ with the flat output given by $\vcy = [\vcx^1_1\quad \vcx^2_1\quad\ldots\quad\vcx^N_1]$.
\end{theorem}

\begin{proof}
    Note that the joint dynamics \eqref{eq: joint-dynamics} is pure-feedback under the state decomposition ${\vcx = [\vcx_1\;\vcx_2\ldots\;\vcx_r]}$, as
    \ifniceformat
    \begin{multline}\label{eq: joint-dynamics-pure-feedback}
        \dot\vcx_k = f_k(\vcx_1, \ldots, \vcx_k, \vcx_{k+1}) :=
        \begin{bmatrix}
            \bar{f}^1_k(\vcx^1_1, \ldots, \vcx^1_k, \vcx^1_{k+1}) + 
                \bar{\Delta}^1_k(\vcx_1, \ldots, \vcx_k)\\
            \bar{f}^2_k(\vcx^2_1, \ldots, \vcx^2_k, \vcx^2_{k+1}) +
                \bar{\Delta}^2_k(\vcx_1, \ldots, \vcx_k, \vcx^1_{k+1}) \\
            \vdots \\
            \bar{f}^N_k(\vcx^N_1, \ldots, \vcx^N_k, \vcx^N_{k+1}) +
            \bar{\Delta}^N_k(\vcx_1, \ldots, \vcx_k, \vcx^{<N}_{k+1})
        \end{bmatrix}.
    \end{multline}
    \fi
    Moreover, the joint pure-feedback system is regular, as
    \begin{equation}\label{eq: joint-dynamics-regular}
        \begin{aligned}
            &\left\lvert D_{\vcx_{i+1}}f_i(\vcx^{*}_1, \ldots, \vcx^{*}_i, \vcx^*_{i+1})\right\rvert\\
            =&\det
            \begin{bmatrix}
                D_{\vcx^1_{i+1}}f^1_i & 0 & \cdots & 0\\
                D_{\vcx^1_{i+1}}\Delta^2_i & D_{\vcx^2_{i+1}}f^2_i & \cdots & 0\\
                \vdots & \vdots & \ddots & \vdots\\
                D_{\vcx^1_{i+1}}\Delta^N_i & D_{\vcx^2_{i+1}}\Delta^N_i & \cdots & D_{\vcx^N_{i+1}}f^N_i
            \end{bmatrix}\\
            =& \prod_{j=1}^{N}\det\left(D_{\vcx^j_{i+1}}f^j_i\right) \neq 0,
        \end{aligned}
    \end{equation}
    where the second equality is a result of the lower-triangular structure of the matrix, and the last equality uses the assumption that the uncoupled individual systems are regular. The flatness of \eqref{eq: joint-dynamics} and its flat outputs then follows directly from the well-known result that regular pure-feedback systems are differentially flat \citep{murray1995differential}.
\end{proof}
Crucially, Theorem~\ref{thm: flatness-of-canonical-coupling} shows that the joint flat output $\vcy$ is simply the concatenation of the individual subsystem flat outputs $\vcy^i$. This will allow us to explicitly construct the flatness diffeomorphism of the joint system from those of the uncoupled systems.

\begin{assumption}\label{assm: nominal-flat-map}
    For all $i=1, ..., N$, a set of smooth functions ${h^i_k: \R^{(k+1)m}\to\R^m}$ that satisfy
    \begin{equation}\label{eq: defn-of-g}
        \begin{aligned}
            &h^{i}_k\left(\vcx^{i}_1, \ldots, \vcx^{i}_{k}, \bar{f}^{i}_k(\vcx^{i}_1, \ldots, \vcx^{i}_{k+1})\right) = \vcx^{i}_{k+1}, \\
            &\hspace{47mm} k = 1, \ldots, r-1,\\
            &h^{i}_{r}\left(\vcx^{i}_1, \ldots, \vcx^{i}_{k}, \bar{f}^{i}_r(\vcx^{i}_1, \ldots, \vcx^{i}_{r}, \vcu)\right) = \vcu
        \end{aligned}
    \end{equation}
    for all ${(\vcx, \vcu) \in \mathcal{X}}$ is known\footnote{The existence of such maps are guaranteed by Assumption~\ref{assm: uncoupled-dynamics-regular} and the implicit function theorem. (See \citet{yang2025learning} for a proof.) \rewrite{}{For many well-known flat systems, they are known and can be expressed in closed-form.}}.
\end{assumption}
%
%
\begin{theorem}\label{thm: form-of-pert-flat-map}
    Let the assumptions of Theorem~\ref{thm: flatness-of-canonical-coupling} and Assumption~\ref{assm: nominal-flat-map} hold. Define the recursively constructed maps $\Phi^i_k$ in terms of the functions $h^i_k$ given in equation~\eqref{eq: defn-of-g} as follows:
    \ifniceformat
    \begin{equation}\label{eq: pert-flat-map-construction}
    \begin{aligned}
        &\Phi^i_1(\vcy) = \vcy^i, \\
        &\Phi^i_k(\vcy,\ldots,\vcy^{(k-1)}) = h^i_{k-1}\Bigg(\Phi^i_1, \ldots, \Phi^i_{k-1},        D\Phi_{k-1}^i
        \begin{bmatrix}
            \dot\vcy \\ \vdots \\ \vcy^{(k-1)}.
        \end{bmatrix}
        - \Delta^i_{k-1}(\Phi_1, \ldots, \Phi_{k-1}, \Phi^1_k, \ldots, \Phi^{i-1}_k)
        \Bigg),
    \end{aligned}
    \end{equation}
    \fi
    for $k=2, \ldots, r+1$, wherein we omit the arguments of the $\Phi^j_k$ on the right-hand side to reduce notational clutter. Then, under the joint dynamics $\bar f + \Delta$ and for all $i=1, \ldots, N$ and $k=1,\ldots,r$, we have
    \begin{equation}\label{eq: pert-map-result}
        \begin{aligned}
            &\vcx^i_k = \Phi^i_k(\vcy, \ldots, \vcy^{(k-1)}), \quad
            &\vcu^i = \Phi^i_{r+1}(\vcy, \ldots, \vcy^{(r)})
        \end{aligned}
    \end{equation}
    for $(\vcx, \vcu) \in \mathcal{X}$.
\end{theorem}
\begin{proof}
    We proceed by induction on the lexicographic index $\iota(i,k) := (k-1) \cdot N + i$, over the domain $\{(i,k) \mid i = 1,\dots,N;\ k = 1,\dots,r+1\}$. For the base case $k=1$, we have from Theorem~\ref{thm: flatness-of-canonical-coupling} that, for each subsystem $i$, the flat output is $\vcy^i = \vcx_1^i$. Thus, $\hat{\Phi}_1^i(\vcy) = \vcy^i = \vcx_1^i$.
    For the sake of induction, assume that for all $(j,s)$ such that $\iota(j,s) < \iota(i,k)$, we have
    \begin{equation*}
        \vcx_s^j = \Phi_s^j(\vcy, \dot{\vcy}, \dots, \vcy^{(s-1)}).
    \end{equation*} 
    It suffices to show that the equality also holds for ${\Phi^i_k}$. Note that \textit{under the joint dynamics}, we have
    \begin{equation*}
        \bar{f}^i_{k-1}(\vcx_1, \ldots, \vcx_{k}) = \dot\vcx^i_{k-1} - \Delta^i_{k-1}(\vcx_{<k}, \vcx^{<i}_k).
    \end{equation*}
    Thus, by the definition of $h$ in \eqref{eq: defn-of-g},
    \begin{equation*}
        h^i_{k-1}\left(\vcx_1, \ldots, \vcx_{k-1}, \dot{\vcx}^i_{k-1} - \Delta^i_{k-1}(\vcx_{<k}, \vcx^{<i}_k)\right) = \vcx^i_k.
    \end{equation*}
    By the induction hypothesis, all arguments on the right-hand side are given by $\Phi_s^j(y, \dots)$ for $\iota(j,s) < \iota(i,k)$. Hence,
    \begin{equation*}
        \vcx^i_{k} = \Phi^i_{k}(\vcy, \ldots, \vcy^{(k-1)}, \vcy^k).
    \end{equation*}
    Since $\Phi_{k'}, \Delta_{k-1}$, and $h_{k-1}$ are smooth, ${\hat{\Phi}_k}$ is also smooth. 
    Adopting the mild abuse of notation $\vcx_{r+1}:= \vcu$, the desired result for $\Psi$ follows from the same argument. 
\end{proof}
\rewrite{}{Under Assumption~\ref{assm: nominal-flat-map}, the inverse maps $h_k^i$ are known in closed form. Consequently, for systems with moderate complexity, the maps $\Phi_k^i$ can often be explicitly derived \textit{offline}, either by hand or via symbolic differentiation, thus incurring minimal computational overhead \textit{online}. For systems with a high relative degree $r$ or highly complex coupling terms, the required time derivatives of $\Delta_{k-1}^i$ can be evaluated online by employing Taylor-mode automatic differentiation (AD), which allows the computational complexity to scale polynomially in $r$ while maintaining numerical stability.}

\section{Structural Properties of the Flatness Diffeomorphism}
\label{sec: locality-of-diffeo}
The existence of a flatness diffeomorphism for the joint system simplifies control design. However, in general, evaluating this joint diffeomorphism may require subsystems to access the flat outputs (and corresponding derivatives) of all other subsystems, which may not be feasible in practice. Fortunately, for lower-triangular couplings, the maps $(\Phi^i, \Psi^i)$ will inherit a sparsity pattern from the coupling $\Delta$, allowing them to be computed with information from only a subset of subsystems. In what follows, we study this sparsity structure and its dependency on the coupling $\Delta$.

We characterize the sparsity of the joint flatness diffeomorphism at each state ${\vcx}$ in terms of the dependencies needed by each subsystem to compute its component of the dynamic coupling $\Delta$ and maps $(\Phi, \Psi)$. Fewer dependencies lead to less information sharing between subsystems, resulting in a more scalable distributed controller. To formalize the sparsity structure of $\Delta$ at joint state $\vcx$, we define a directed coupling graph $\stG(\vcx) = (\stV, \stE(\vcx))$, where $\stV = \{1,\ldots,N\}$ indexes the subsystems. An edge $(j, i)$ is included in the state-dependent edge set $\stE(\vcx)$ if the $i^\textrm{th}$ subsystem's dynamics depend on the $j^\textrm{th}$ subsystem's state at $\vcx$, i.e., that
\[
(j, i) \in \stE(\vcx) \iff \exists\, k\in[r], \quad D_{\vcx^j}^k \Delta^i(\vcx) \ne 0 .
\]
Importantly, the sparsity of the dependency graph is not fixed but varies with the state. Unlike a static dependency graph that captures all possible couplings on $\stX$, $\stG(\vcx)$ reflects only the active dependencies at $\vcx$, 
potentially leading to less information exchange needed between subsystems. Importantly, we note that while we formulated the edge set as state-varying to best capture its sparsity, we still assume that the acyclic state-ordering is fixed and coincides with the system indices. This might arise, for example, in the case of quadrotors coupled under downwash, where the quadrotors have a fixed ordering of their altitude, but can move in and out of the effective region of each other's downwash.

In what follows, denote the set of $k$-hop in-neighbors of node $v$ by $\stN_{\stG(\vcx)}^k(v)$, such that $\stN^0_{\stG(\vcx)}(v) := \{v\}$ and
\[
\stN^k_{\stG(\vcx)}(v)=\{j \in \stV \mid \exists\ i \in \stN^{k-1}_{\stG(\vcx)}(v)\ s.t.\ (j, i) \in \stE\}
\]
for $k\geq 1$. The ancestor set of $v$, denoted $\mathbf{Anc}_{\stG(\vcx)}(v)$, is the set of all nodes from which there exists a directed path to $v$, i.e., 
\begin{equation}
    \mathbf{Anc}_{\stG(\vcx)}(v) = \bigcup_{k \geq 0} \stN_{\stG(\vcx)}^k(v).
\end{equation}
Denote by $\stS^i_k(\vcx)$ the set of indices of subsystems whose flat outputs $\{\vcy^j\}_{j \in \mathcal{S}^i_k(\vcx)}$ (and time derivatives) are required to evaluate $\Phi^i_k$ as defined in \eqref{eq: pert-flat-map-construction}. In what follows, we give bounds on $\stS^i_k(\vcx)$ in terms of the structure of $\stG(\vcx)$.

\begin{corollary}
\label{cor: coupling-locality}
Under the assumptions of Theorem~\ref{thm: form-of-pert-flat-map},  
${\mathcal{S}^i_k(\vcx) \subseteq \mathbf{Anc}_{\stG(\vcx)}(i)}$ for all $i\in [N]$ and $k \in [r+1]$.
\end{corollary}
\begin{proof}[Proof]
We proceed by induction on $k$. The base case $k = 1$ holds trivially as $\Phi^i_1 = \vcy^i$ depends only on the flat output of subsystem $i$, i.e., that $\stS^i_1 = \{i\} \subseteq \mathbf{Anc}(i)$\footnote{We omit the dependence of $\mathbf{Anc}_{\stG(\vcx)}$ on $\stG(\vcx)$ and $\stS^i_k(\vcx)$ on $\vcx$ to reduce notational clutter in this proof and the proof for Corollary~\ref{cor: block-coupling-locality}.}. For the sake of induction, assume that $\stS^i_{k'} \subseteq \mathbf{Anc}(i)$ for all $k' < k$. We will show that the condition holds for $\stS^i_{k}$.

Recall the definition~\eqref{eq: pert-flat-map-construction} of the map $\Phi^i_k$. By the inductive hypothesis, the terms $\Phi^i_{1}, \dots, \Phi^i_{k-1}$ and $D\Phi^i_{k-1}[y, \dots, y^{(k-1)}] = \frac{d}{dt}\Phi^i_{k-1}$ depend only on flat outputs of subsystems in $\stS^i_{k-1} \subseteq \mathbf{Anc}(i)$. From the definition of $\stG$, the term $\Delta^i_{k-1}$ only depend on subsystems $j$ where $(j, i) \in \stE$. We now analyze the dependencies of its arguments: lower order maps\footnote{We use the notation $\Phi^j_{<k}$ and $\Phi^{<i}_k$ as a parallel of the notation in \eqref{defn: lower-triangular-coupling}.} $\Phi^j_{<k}$ and the $k$-th order lower-indexed maps $\Phi^{<i}_k$. First, from the induction hypothesis, for any $k'<k$, $\Phi^j_{k'}$ depends on subsystems $\stS^j_{k'} \subseteq \mathbf{Anc}(j)\subseteq \mathbf{Anc}(i)$, since $(j, i) \in \stE$. To handle the same-order terms $\Phi^j_k$ with $j < i$, we apply a secondary induction on subsystem index $j$. For $j = 1$, there are no earlier subsystems, so $\stS^1_k \subseteq \mathbf{Anc}(1)$. Assume for induction that $\stS^j_k \subseteq \mathbf{Anc}(j)$ holds for all $j < i$. $\Delta^i_{k-1}$ only depends on $\Phi^j_k$ for $j < i$ and $(j, i) \in \stE$. Thus, by the induction hypothesis, we have $\stS^i_k \subseteq \bigcup_{j<i\wedge (j,i) \in \stE}\mathbf{Anc}(j) \subseteq \mathbf{Anc}(i)$, concluding inner inductions. Therefore, $\Phi^i_k$ depends only on the stated set $\mathcal{S}^i_k$, completing the outer induction.
\end{proof}

Corollary~\ref{cor: coupling-locality} shows that for subsystem $i$, evaluating the map $\Phi^i_k$ at $\vcx$ may require flat output derivatives for all ancestor subsystems in the dependency graph, $\mathbf{Anc}_{\stG(\vcx)}(i)$. This is similar to the structure observed in \citet{lamperski2015optimal} for linear distributed optimal control with directed acyclic graph structure. While this significantly reduces the information requirement for systems of lower indices, the ancestor set of the final subsystem $N$ could encompass the entire network, thereby necessitating global information. However, the result can be strengthened when the coupling \eqref{eq: defn-canonical-lt-coupling} does not depend on  $\vcx^{<i}_{j+1}$, in which case any subsystem $i$ requires only information from its $(k-1)$-hop in-neighbors to compute $\Phi^i_k$.

\begin{definition}[Strongly Lower-Triangular Coupling]\label{defn: block-lower-triangular}
    We say that the coupling term $\Delta(\vcx)$ is \textit{strongly} lower-triangular on $\stX$ if it is lower-triangular, and its decomposition $\Delta^i_j$ can be expressed as
    \begin{equation}\label{eq: defn-block-lt-coupling}
        \Delta^i_j(\vcx) = \bar{\Delta}^i_j(\vcx_1, \ldots, \vcx_{j})
    \end{equation}
    for all $(\vcx, \vcu) \in \stX$.
\end{definition}

\begin{corollary}
\label{cor: block-coupling-locality}
Let the assumptions of Theorem~\ref{thm: form-of-pert-flat-map} hold and suppose that the coupling $\Delta(\vcx)$ is \textit{strongly} lower-triangular. Then,
$\stS^i_k(\vcx) \subseteq \stN_{\stG(\vcx)}^{k-1}(i)$.
\end{corollary}
\begin{proof}
Similar to Corollary~\ref{cor: coupling-locality}, we proceed by induction on $k$, with the base case $\stS^i_1 = \{i\} \subseteq \stN^{0}(i)$. For the sake of induction, assume that $\stS^i_{k'} \subseteq \stN^{k'-1}(i)$ for all $k' < k$. We will show that the condition holds for $\stS^i_{k}$.

By a similar argument, we have that the terms $\Phi^i_{1}, \dots, \Phi^i_{k-1}$ and $D\Phi^i_{k-1}[y, \dots, y^{(k-1)}] = \frac{d}{dt}\Phi^i_{k-1}$ in \eqref{eq: pert-flat-map-construction} depend only on flat outputs of systems in $\stS^i_{k-1} \subseteq \stN^{k-1}(i) \subseteq \stN^{k}(i)$. By the definition of strongly lower-triangular coupling, $\Delta^i_{k-1}$ only depends on substates $\vcx^j_{<k}$ for $(j, i) \in \stE$. By the inductive hypothesis, each $\vcx^j_{k'}$ with $k' < k$ is reconstructed by $\Phi^j_{k'}$ with dependence on $\stS^j_{k'} \subseteq \stN^j_{k'-1}$. Since $(j, i) \in \stE$, we have $\stN^j_{k'-1} \subseteq \stN^i_{k'} \subseteq \stN^i_{k-1}$. Therefore, the dependency of $\Phi^i_k$ is confined to $\mathcal{N}_{k-1}^i$, thereby concluding the induction.
\end{proof}

The corollary implies that when the coupling is strongly lower-triangular, the flatness diffeomorphism of subsystem $i$ can be computed using at most the flat outputs (and their derivatives) of its $r$-hop neighboring systems. When $\stG(\vcx)$ is sparse, this leads to efficient and scalable distributed implementations as the size of this neighborhood is independent of the total number of subsystems $N$.

\section{Distributed Tracking Controller}\label{sec: tracking-control}
We now design a distributed controller for trajectory tracking. To do so, we emphasize that the dynamics of the flat outputs, $\vcy = \vcx_1$, (i.e., the flat dynamics) decouple across the subsystems. Specifically, the flat dynamics of subsystem $i$ take the form
\begin{equation}\label{eq: unperturbed-flat-dynamics}
    \dot\vcz^i = A^i\vcz^i + B^i\vcv^i, \qquad i = 1, \ldots, N,
\end{equation}
where $\vcz^i := [\vcy^i, \dot\vcy^i, \ldots, (\vcy^i)^{(r)}]^\top$ is the flat state of subsystem $i$, $\vcv^i$ is a virtual input. The matrices $A^i \in \R^{rm_i \times rm_i}$ and $B^i \in \R^{rm_i \times m_i}$ are chosen to match the Brunovsky normal form, i.e., that ${(\vcy^i)}^{(r+1)} = \vcv^i$. Let $(\vcz_{\mathrm{ref}}^i(t), \vcv_{\mathrm{ref}}^i(t))$ denote a feasible flat state and virtual input trajectory for subsystem $i$ that satisfies the flat dynamics \eqref{eq: unperturbed-flat-dynamics}. Assuming that each subsystem $i$ can accurately measure its own flat state $\vcz^i$, our goal is to design a distributed tracking controller that can generate the physical input $\vcu^i(t)$ to stabilize the corresponding physical reference trajectory while minimizing information exchange between the subsystems. To do so, we first design a fully decentralized error feedback controller for the linear flat dynamics of each system. We then use the flatness diffeomorphism to convert the virtual input into nominal input based on flat state measurement.

Defining the tracking error $\tilde{\vcz}^i:= \vcz^i - \vcz^i_{\mathrm{ref}}$, we propose the following error-feedback controller for the flat dynamics:
\begin{equation}\label{eq: feedback-tracking-controller}
\vcv^i(t) = \vcv_{\mathrm{ref}}^i(t) - K^i \tilde{\vcz}^i(t),
\end{equation}
where $K^i \in \R^{m_i \times rm_i}$ is a stabilizing feedback gain, which can be chosen, e.g., via pole placement. We can then compute the control input $\vcu^i$ by evaluating the flatness diffeomorphism $\Psi^i$, which depends on flat state $\vcz^i$, virtual input $\vcv^i$, and $(\vcz^j, \vcv^j)$ of an appropriate subset of the other subsystems. By the results in Section~\ref{sec: locality-of-diffeo}, subsystem $i$ needs to gather $(\vcz^j, \vcv^j)$ from subsystems $j \in \mathbf{Anc}_{\stG(\vcx)}(i)(\vcx)$ to evaluate $\Psi^i$. Further, when the coupling is strongly lower-triangular, this dependency is further localized to its $r$-hop neighbors, $\stN^r_{\stG(\vcx)}(i)$. \rewrite{}{In the prior case, discovering and routing data through a dynamic, global ancestor graph online can be difficult; thus, practical implementation might require the communication graph to be known and established a priori. However, in the latter case, since the coupling structure often reflects spatial proximity, agents can dynamically discover these dependencies online using localized peer-to-peer communication. Further,} this allows for efficient distributed implementation as the information sharing between the systems does not scale with the number of subsystems $N$. The overall pipeline is summarized in Algorithm~\ref{alg: tracking-alg}.
\begin{algorithm}
    \caption{Controller for Subsystem $i$}\label{alg: tracking-alg}
    \KwIn{Reference trajectory $(\vcz^i_{\mathrm{ref}}(t), \vcv^i_{\mathrm{ref}}(t))$}
    \KwOut{Physical control input $\vcu^i(t)$}
    Measure current flat state: $\vcz^i(t)$\;
    Compute error: $\tilde{\vcz}^i(t) \gets \vcz^i(t) - \vcz^i_{\mathrm{ref}}(t)$\;
    Compute virtual input: $\vcv^i(t) \gets \vcv^i_{\mathrm{ref}}(t) - K^i \tilde{\vcz}^i(t)$\;
    Gather $(\vcz^j(t), \vcv^j(t))$ for $j \in \stS^i_{r+1}(\vcx)$\;
    Evaluate $\vcu^i(t) \gets \Psi^i(\vcz, \vcv)$\;
\end{algorithm}
\section{Experiments}\label{sec: experiments}
We demonstrate our algorithm by synthesizing distributed tracking controllers for a swarm of 2D quadrotors subject to downwash interactions. We consider a downwash model that satisfies the lower-triangular structure and show that our proposed framework can synthesize flatness-based distributed controllers with high tracking performance. Then, we show how to approximate the coupling with a strongly lower-triangular surrogate, which leads to controllers that can operate under reduced communication at the cost of a slight increase in tracking error. The code and animations for our experiments can be found at \url{https://github.com/FJYang96/flat-distributed}.

\subsection{System Definition}
We consider each quadrotor $i$ as a subsystem, with states $\vcx^i \in \R^{8}$ and control inputs $\vcu^i \in \R^{2}$, where 
\begin{equation*}
    \vcx_1^i = \vcp^i, \vcx^i_2=\vcv^i, \vcx^i_3 = [T^i, \theta^i], \vcx^i_4 = [\dot T^i, \omega^i], \vcu^i = [\ddot T^i, \tau^i]
\end{equation*}
and the subsystem dynamics are given by
\begin{alignat}{3}
        &\dot \vcx^i_1 = \bar{f}_1^i(\vcx^i_2) = \vcx^i_2, \ \  &&\dot\vcx^i_2 = \bar{f}_2^i(\vcx^i_3) &&= 
        \frac{1}{\mathfrak m}\begin{bmatrix} -T^i\sin\theta^i 
        \\ T^i\cos\theta^i - \mathfrak{m}^ig \end{bmatrix}, \nonumber \\
        &\dot \vcx^i_3 = \bar{f}_3^i(\vcx^i_4) = \vcx^i_4, \ \  &&\dot \vcx^i_4 = \bar{f}_4^i(\vcu^i) &&= \begin{bmatrix} 1 & 0 \\ 0 & I^{-1} \end{bmatrix}\vcu^i.
        \label{eq: 2dquad}
\end{alignat}
Here, $\vcp^i, \vcv^i \in \R^2$ are the position and velocity; $\theta^i \in [-\pi, \pi)$ and $\omega^i \in \R$ are the orientation and angular velocity; and $T^i, \tau^i \in \R$ are the total thrust and torque. The mass $\mathfrak{m}$ and moment of inertia $I$ are positive constants, and $g$ denotes gravitational acceleration. The thrust $T^i$ is dynamically extended twice to cast the subsystem into pure-feedback form, and the flatness diffeomorphism is detailed in \S4 of \citet{yang2025learning}.

The dynamics of individual quadrotors are coupled by downwash-induced drag, where a vehicle flying higher exerts a downward force on those below it. We model the downward drag force experienced by vehicle $i$ as a sum of all pairwise interactions from vehicles above:
\begin{equation}\label{eq: downwash-coupling}
\begin{aligned}
    \dot\vcv^{i}_{2}
    &=\bar{f}^i_{2,2}(\vcx^i_3) +
    \frac{1}{\mathfrak{m}^i}\sum_{(j, i) \in \stE(\vcx)} F_D(\vcp^j - \vcp^i, T^j),
\end{aligned}
\end{equation}
where $\stG(\vcx)=(\stV, \stE(\vcx))$ is a directed acyclic graph, where $(j, i) \in \stE(\vcx)$ if and only if $j$ is above $i$, i.e., that $\vcp^j_2 > \vcp^i_2$. Here, the force on $i$ depends only on its relative position to higher vehicles and their thrust. We note that any smooth function $F_D$ with this dependency yields a lower-triangular coupling (Def.~\ref{defn: lower-triangular-coupling}), ensuring that the coupled system admits the flatness diffeomorphism of Theorem~\ref{thm: form-of-pert-flat-map}.

For the following experiments, we adopt a form of $F_D$ inspired by \citet{jain2019modeling} and simplified for the 2D setup under near-hover flight conditions.
\ifniceformat
In the interest of brevity, we defer the exact derivation, expression, and parameters to Appendix~\ref{sec: downwash-model}.
\fi
Notably, $F_D$ decays with both horizontal and vertical separations and is proportional to the thrust of the upper vehicle. We visualize its decay with respect to vehicle separation in Figure~\ref{fig: downwash-viz}.

\begin{figure}
    \centering
    \ifniceformat
        \includegraphics[width=0.6\linewidth]{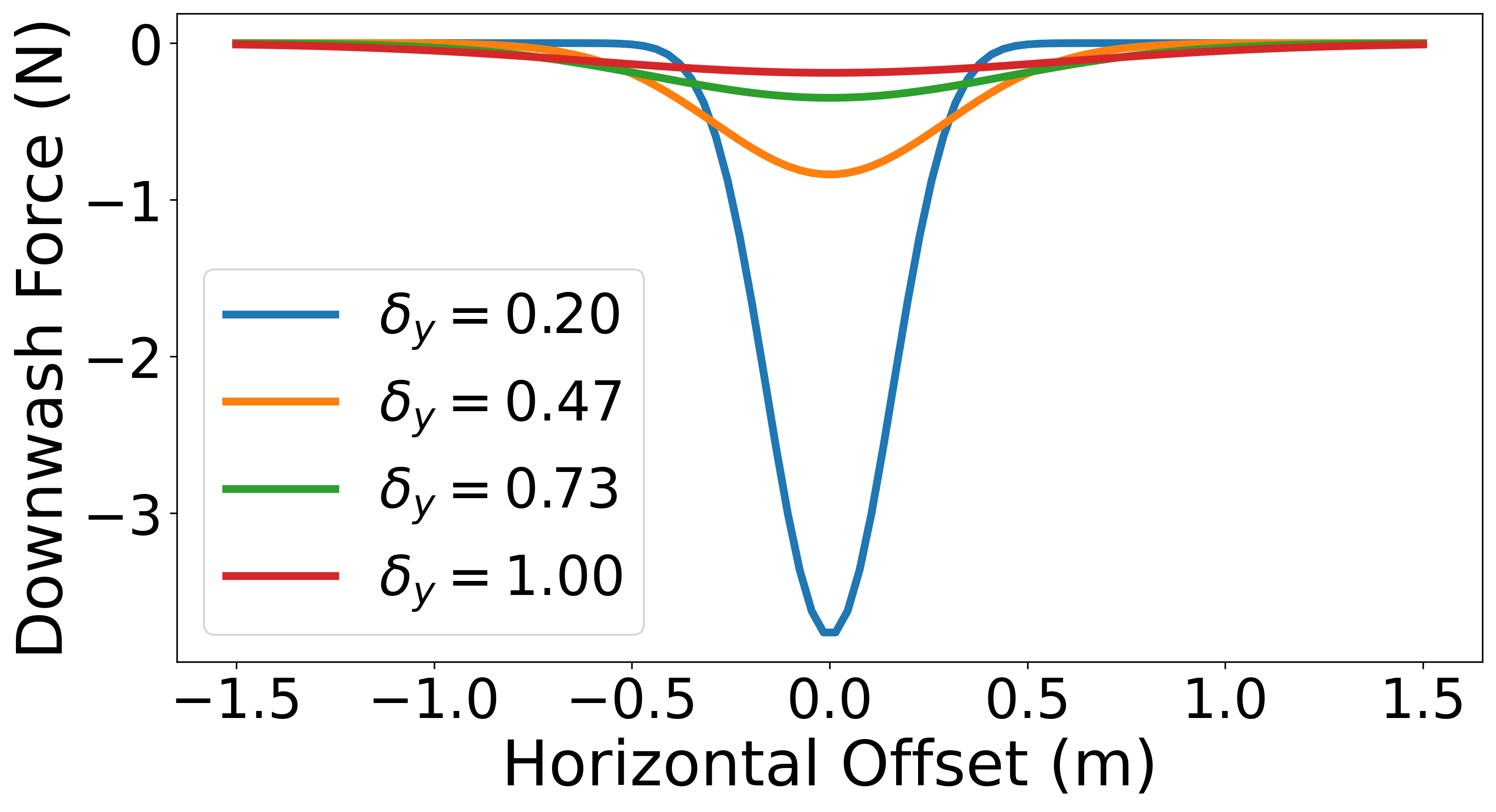}
    \fi
    \caption{Visualization of the $F_D$ for a pair of vehicles as a function of the horizontal separation, plotted at different vertical separations $\delta_y$.}
    \label{fig: downwash-viz}
\end{figure}

\subsection{Experimental Scenario}
We evaluate the tracking controllers from Section~\ref{sec: tracking-control} in a simulation with a group of quadrotors subject to the downwash model. The reference trajectory divides the $N$ quadrotors equally into two stacked formations, where the two groups fly past each other laterally at a constant altitude. Specifically, the reference trajectory (namely, the desired horizontal and vertical position) for quadrotor $i$ is given by
\begin{equation*}
    \vcp^i_{ref,1}(t) = (-1)^{i}\left(\frac{vT}{2} - vt\right),\quad \vcp^i_{ref,2}(t) = N + 1 - i.
\end{equation*}
with a duration $T=5 \, \text{s}$ and speed $v=1 \, \text{m/s}$.

We assume that the controller for each vehicle receives perfect measurements of position, velocity, and acceleration. Higher-order derivatives of the flat output are estimated by a state observer. The dynamics are integrated using a fourth-order Runge-Kutta method with a time step of $\tau=0.01 \, \text{s}$, and the control input is applied via a zero-order hold at $100 \, \text{Hz}$. As the measure of tracking performance, we compute the \textit{mean} position error, $\vce_{\text{pos}}$, averaged over all vehicles and the trajectory duration, given as
\begin{equation*}
    \vce_{pos} = \frac{\tau}{NT}\sum_{i=1}^{N}\sum_{k=1}^{\lfloor T/\tau \rfloor} \norm{\vcx^i_1(k\tau) - \vcp^i_{ref}(k\tau)}.
\end{equation*}

\subsection{Implementations of Distributed Flatness-Based Control}
For a swarm of $N=4$ quadrotors, we implement three different distributed flatness-based controllers based on Algorithm~\ref{alg: tracking-alg}. Across these implementations, we vary the accuracy of the model of dynamic coupling used to construct the flatness diffeomorphism (and hence, $\Psi^i$), despite applying all controllers to the coupled dynamics in \eqref{eq: downwash-coupling}. Control gains $K^i$ for the linear controller in the flat space are chosen to be the optimal LQR gains for costs $\mtQ = 100 \mtI_8$ and $\mtR = \mtI_2$.

\subsubsection{Exact Controller}
We first construct a tracking control algorithm where the map $\Psi^i$ is constructed by following Theorem~\ref{thm: form-of-pert-flat-map} with $\Delta$ corresponding to the \textit{exact} downwash model \eqref{eq: downwash-coupling}, which is lower-triangular. 
Note that by Corollary~\ref{cor: coupling-locality}, evaluating the flat output transformation $\Psi^i$ for subsystem $i$ requires information from its ancestor set $\mathbf{Anc}_\stG(i)$, which corresponds to all vehicles flying at a greater altitude.

\subsubsection{Approximate Controller}
To further reduce communication demands, we develop an approximate downwash model with two simplifications. First, considering the decay of the downwash force with increased separation, we ignore the influence from vehicles sufficiently far away. Second, we use the near-hover assumption to approximate the upper vehicle thrust in $F_D$ as its weight. This leads to the \textit{approximate} downwash model
\begin{equation}\label{eq: approx-downwash-coupling}
    \dot\vcv^{i}_{2} =
    \bar{f}^i_{2,2}(\vcx^i_3) +
    \sum_{(j, i) \in \hat\stE(\vcx)} \hat{F}_D(\vcp^j - \vcp^i),
\end{equation}
where $\hat F_D(\vcdelta) = F_D(\vcdelta, \mathfrak{m}g)$. The approximate dependency graph $\hat\stG$ contains edge $(j, i)$ if $(j, i)\in\stE(\vcx)$ and ${-\bar\vcdelta \prec\vcp^j - \vcp^i \prec \bar{\vcdelta}}$ for a user-picked threshold $\bar\vcdelta\in\R^2_+$. For the following comparison, we set $\bar\vcdelta = [0.5, 2.5]$. Crucially, the approximate downwash model is \textit{strongly} lower-triangular (Def.~\ref{defn: block-lower-triangular}), and moreover, the coupling only affects the vertical acceleration and no states of other orders. Following an analysis in the same fashion as Corollary~\ref{cor: block-coupling-locality}, to evaluate $\Psi^i$, vehicle $i$ only needs information from $\stN^1_{\hat\stG}$, i.e., other vehicles that are within the threshold range.

\subsubsection{Nominal Controller}
As a baseline, we also implement the distributed tracking controller under the assumption of \textit{nominal} conditions (i.e., no coupling), wherein $\Psi^i = \bar\Psi^i$ (from the flatness diffeomorphism of the uncoupled subsystem $i$). This controller does not model the downwash effects at all and relies on each subsystem's controller to reject the downwash effects as an external disturbance. 

\subsection{Experimental Results}

We show the evolution of the position error over time for all controllers in Figure~\ref{fig:tracking-error}. It is clear that downwash has significant effects on the quadrotors' dynamics, as the \textit{Nominal} controller (which ignores all downwash effects) suffers from a high mean position error of 0.0572. In contrast, the \textit{Exact} controller rejects virtually all downwash disturbance, achieving a near-zero mean position error of 0.0001. This corroborates the diffeomorphism construction from Theorem~\ref{thm: form-of-pert-flat-map} and the ability of our framework to leverage this flatness property to synthesize distributed controllers for coupled systems. On the other hand, the \textit{Approximate} controller significantly outperforms the \textit{Nominal} controller, with a mean position error of 0.015. This demonstrates its effectiveness as a trade-off, reducing the communication burden of the \textit{Exact} controller while still providing adequate compensation. An illustration of the trajectories is provided in Figure~\ref{fig:trajectory}.

\begin{figure}
    \centering
    \ifniceformat
        \includegraphics[width=0.6\linewidth]{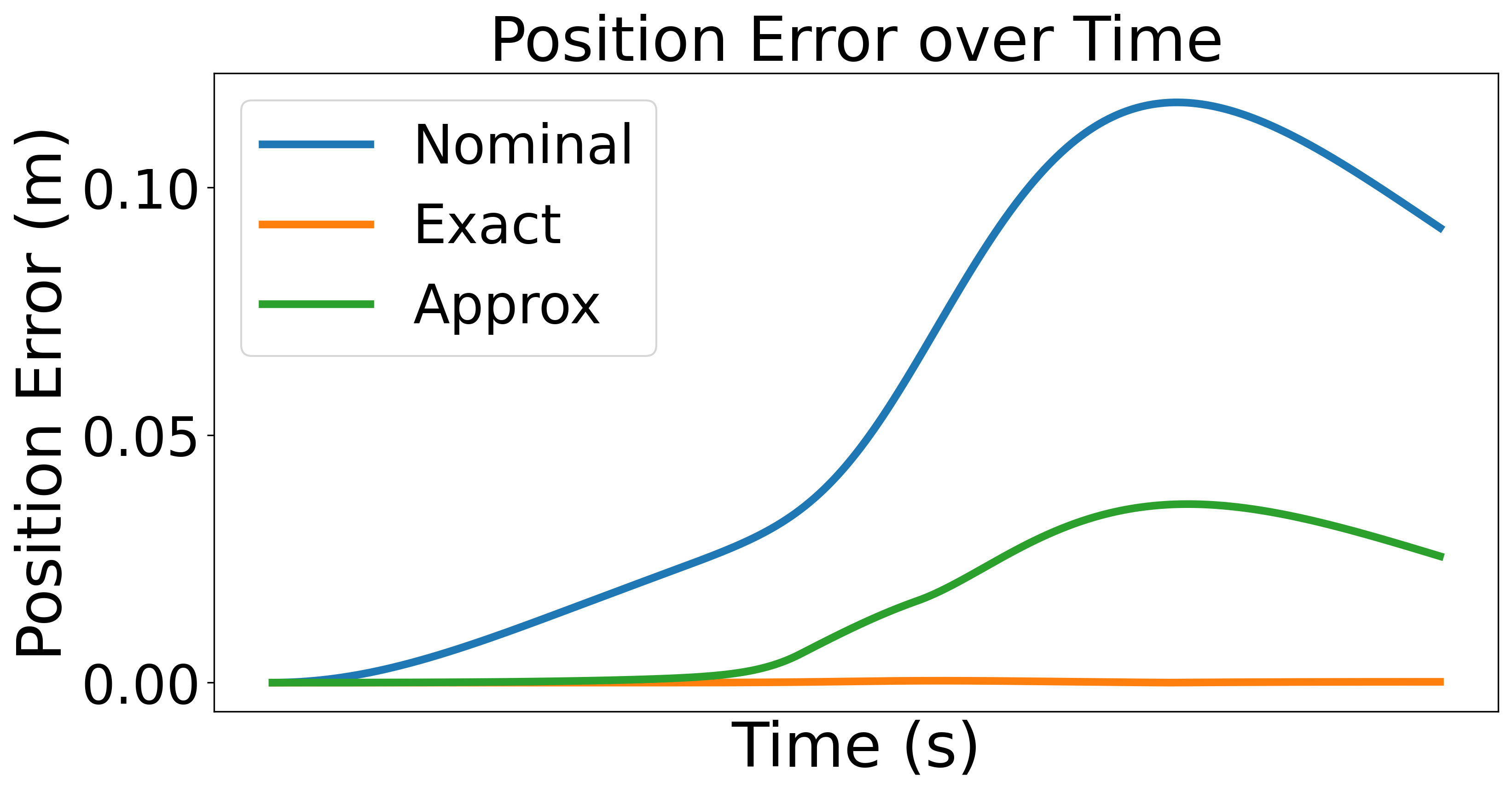}
    \fi
    \caption{Evolution of position error over time, averaged over the vehicles.}
    \label{fig:tracking-error}
\end{figure}

\begin{figure}
    \centering
    \ifniceformat
        \includegraphics[width=0.6\linewidth]{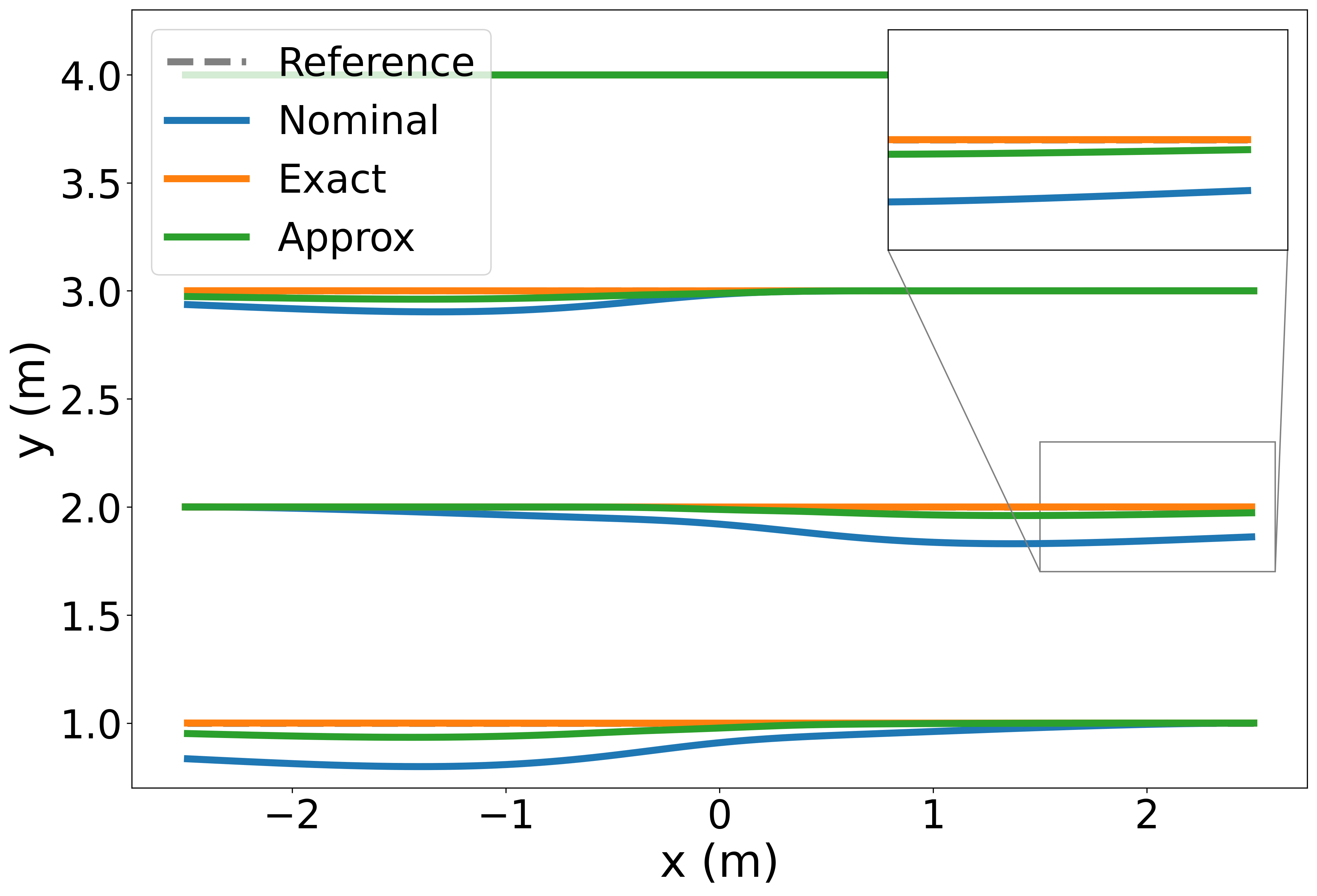}
    \fi
    \caption{Trajectory of the quadrotors. Dotted grey lines represent the reference trajectories. Flatness-based controllers constructed under both the exact and the approximate downwash coupling achieve good tracking performance.}
    \label{fig:trajectory}
\end{figure}

Finally, we investigate how the approximation quality varies as we vary the approximation range $\bar\vcdelta$. For the following experiments, we set $\bar\vcdelta = [\bar\delta, \bar\delta]$ and vary the constant $\bar\delta$ for $N=10$ vehicles in the same setup as the previous experiment. We show the results in Figure~\ref{fig: approximation-threshold}. As we increase the cutoff threshold, we observe better tracking performance (lower position error) at the cost of increased information exchange between subsystems. This further demonstrates the applicability of the approximation method to achieve a desired trade-off between information sharing and tracking performance.

\begin{figure}
    \centering
    \ifniceformat
        \includegraphics[width=0.6\linewidth]{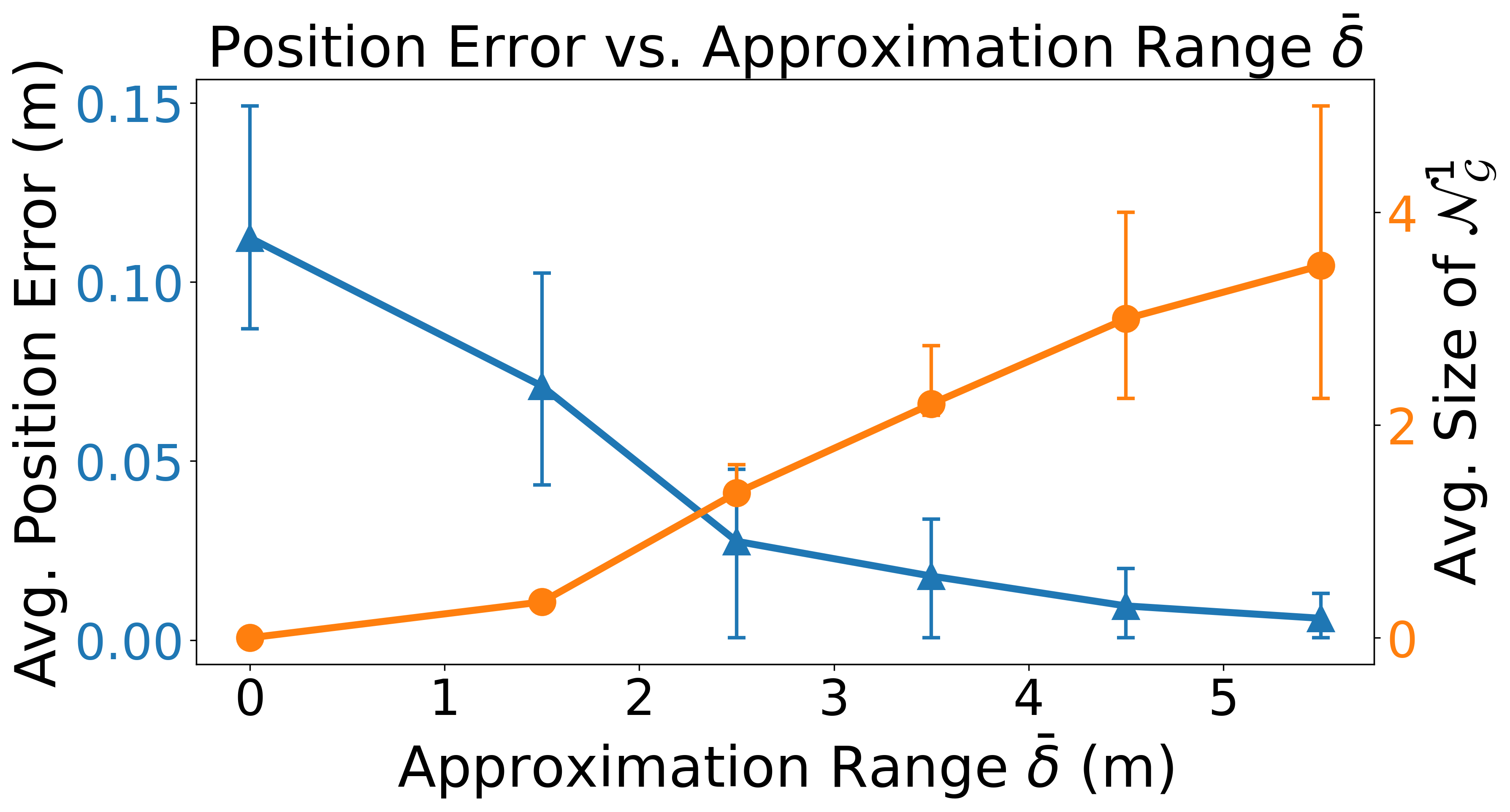}
    \fi
    \caption{Average position errors for $N=10$ vehicles as we vary the approximation cutoff threshold $\bar\vcdelta=[\bar\vcdelta, \bar\vcdelta]$. The position error decreases as we increase the cutoff threshold (communication range), as the average number of vehicles in the communication set increases. The error bars denote the inter-quartile range computed over the vehicles.}
    \label{fig: approximation-threshold}
\end{figure}

\section{Conclusion}
We proposed a distributed controller for networks of flat, pure-feedback systems under dynamic coupling. We identified a class of lower-triangular couplings that preserve the flatness for the joint system and allow the joint flatness diffeomorphism to be explicitly constructed. We analyzed and exploited the sparsity pattern of the flatness diffeomorphism to design distributed controllers under communication constraints. We validated the framework in simulation for planar quadrotors coupled by downwash drag, showing that our flatness-based distributed controller achieves accurate tracking. Future work includes hardware experiments as well, relaxing the lower-triangular constraints on flatness-preserving couplings, \rewrite{}{and incorporating adaptive control techniques to reject disturbances that arise from inaccuracies in approximate coupling models \citep{join2024flatness}}.


\ifniceformat
\bibliographystyle{unsrtnat}
\fi
\bibliography{bibFiles/sample}

\newpage
\appendix
\ifniceformat
\newcommand{\apdxsection}[1]{\section{#1}}
\newcommand{\apdxsubsection}[1]{\subsection{#1}}
\fi




\apdxsection{Downwash Model}\label{sec: downwash-model}
\apdxsubsection{Derivation of the Downwash Model}
We start from the downwash model proposed in \cite{jain2019modeling} and adapt it to our 2D case under the near-hover assumption. From \cite{jain2019modeling}, one can approximate the magnitude of the wind velocity $V$ below a hovering vehicle as
\begin{equation*}
    V(\vcdelta, T) = C_1\frac{\sqrt{T}L}{\delta_y} \exp\left(-C_2\left(\frac{\delta_x}{\delta_y}\right)^2\right),
\end{equation*}
where $\vcdelta = \vcp^\text{top} - \vcp^\text{bottom} =(\delta_x, \delta_y)$ are the radial (horizontal) and axial (vertical) separation between the two vehicles, respectively. We assume that $\delta_y > 0$, i.e., that the bottom vehicle is strictly below the top. $T$ denotes the thrust force of the top vehicle, while $L$ is the size of both vehicles (double the arm length). $C_1$ and $C_2$ are constants that depend on the geometry of both the propellers and the vehicle. We now proceed to find the direct drag force of this velocity field on the airframe of a vehicle below. While the downwash also alters the inflow of the propellers of the vehicle below and creates a loss of thrust, we assume that this effect is dominated by the airframe drag, which is true when the axial separation is small.

The induced velocity field results in both a drag force and a drag torque on the airframe of the bottom vehicle. Assuming that the airframe of the bottom vehicle has a uniform drag coefficient $C_D$ and is near hover, the \textit{drag density} on the airframe is given by
\begin{align*}
  D(\vcdelta, T)
  &= \frac{1}{2}\rho C_D V^2(\vcdelta, T) \\
  &= \underbrace{\frac{1}{2}\rho C_D C_1^2}_{C_3} \frac{T L^2}{\delta_y^2} \exp\left[-2 C_2\left(\frac{\delta_x}{\delta_y}\right)^2\right]
\end{align*}

The drag force $F_D$ can be found by integrating the drag density over the airframe.
Let $l\in[-L/2,L/2]$ denote the span-wise coordinate ($l=0$ at the hub) and define $k:=2C_{2}/\delta_y^{2}$. From the near-hover assumption, the drag force is
\begin{align*}
    F_D
    &= -C_3 \frac{T L^2}{\delta_y^2} \int_{-L/2}^{L/2}\exp(-k(\delta_x+l)^2) \mathrm dl, 
     \\
     &= -C_3 \frac{\sqrt\pi L^2 T}{2\sqrt{2C_2} \delta_y} \left[
           \erf\left(\sqrt{\tfrac{2C_2}{\delta_y^2}}\,(\delta_x+\tfrac{L}{2})\right)-
           \erf\left(\sqrt{\tfrac{2C_2}{\delta_y^2}}\,(\delta_x-\tfrac{L}{2})\right)
         \right]
    \label{eg:FD_final}
\end{align*}
which evaluates to the form we have in the experiment section by grouping constants appropriately. The drag torque can be found similarly as 
\begin{align*}
    \vctau_D 
    &= \int_{L/2}^{L/2} -l D(\vcdelta, T)\mathrm dl \\
    &= - C_3 \frac{TL^2}{\delta_y^2} \int_{-L/2}^{L/2} l\exp\left[-k(\delta_x+l)^2\right]\mathrm dl\\
    &=C_3 \frac{T L^2}{\delta_y^2} \Bigl\{
        \frac{\delta_y^2}{4C_{2}}
        \Bigl[
            e^{-k(\delta_x-L/2)^{2}}-e^{-k(\delta_x+L/2)^{2}}
        \Bigr] \\
    &\qquad
        + \frac{\delta_x\sqrt{\pi}\,\delta_y}{2\sqrt{2C_{2}}}
        \Bigl[
            \erf\bigl(\sqrt{k}\,(\delta_x-L/2)\bigr)-
            \erf\bigl(\sqrt{k}\,(\delta_x+L/2)\bigr)
        \Bigr]
    \Bigr\}.
\end{align*}

We note that the drag torque also satisfies Definition~\ref{defn: lower-triangular-coupling} and thus our approach applies to reject the torque coupling as well. We did not include this in the experiment section for the sake of simplicity of presentation, but we show videos of our approach applied to the case with drag torques on the GitHub page containing our experiment code.

The parameters for the quadrotor and downwash used for the experiments are reported in Table~\ref{tab: quad-params}.
\begin{table}[ht]
    \centering
    \begin{tabular}{|c|c|}
        \hline
        \textbf{Parameter} & \textbf{Value} \\ \hline
        $\mathfrak{m}$ & 1.0 \\ \hline
        $I$ & 0.1\\ \hline
        $g$ & 9.81 \\ \hline
        $C_1$ & 1.0 \\ \hline
        $C_2$ & 0.7 \\ \hline
        $C_D$ & 1.18 \\ \hline
        $L$ & 0.3 \\ \hline
    \end{tabular}
    \caption{Quadrotor and Disturbance Parameters for Numerical Experiments}
    \label{tab: quad-params}
\end{table}

\end{document}